\documentclass{amsart}
\usepackage{graphicx} 
\usepackage{hyperref}
\usepackage{amsmath,amsfonts,amssymb,amsthm,mathrsfs,tikz-cd, enumerate,mathtools,amscd,tikz-cd}

\DeclarePairedDelimiterX\set[1]\lbrace\rbrace{#1}

\theoremstyle{plain}
\newtheorem{definition}{Definition}[section]

\newtheorem{theorem}{Theorem}[section]

\newtheorem{lemma}{Lemma}[section]

\newcommand{\ra}{\rightarrow}

\newcommand{\RR}{{\mathbb R}}
\newcommand{\NN}{{\mathbb N}}
\newcommand{\ZZ}{{\mathbb Z}}
\newcommand{\CC}{{\mathbb C}}

\newcommand{\QQ}{{\mathbb Q}}

\newcommand{\CF}{{\mathcal F}}

\newcommand{\CT}{{\mathcal T}}

\newcommand{\A}{{\mathsf M}}

\newcommand{\bphi}{{\boldsymbol{\varphi}}}
\newcommand{\bq}{{\boldsymbol{q}}}

\newcommand{\Mat}{{\rm Mat}}
\newcommand{\End}{{\rm End}}

\newcommand{\mm}{{\mathfrak m}}
\newcommand{\Gal}{{\rm Gal}}

\newcommand{\GL}{{\rm GL}}

\title{On thermalization in many-body classical Floquet systems}
\author{Anton Kapustin}
\address{California Institute of Technology}

\begin{document}

\begin{abstract}
    It is expected that a generic closed many-body system prepared in a well-behaved initial state and subjected to a periodic drive will eventually thermalize, i.e. approach the state of  maximal entropy. This property, while compatible with and even demanded by the physical intuition, is much stronger than ergodicity or mixing and is difficult to justify mathematically. We describe an infinite set of classical many-body Floquet systems of algebraic origin for which thermalization of very general initial states can be proved. For example, we show that a Gibbs state of any sufficiently uniform local differentiable Hamiltonian heats up to infinite temperature at long times. We show that in agreement with the physical intuition, the only obstruction to thermalization is the existence of local observables which are periodic in time. 
\end{abstract}

\maketitle

\section{Introduction}

The belief that a generic state of a generic closed system thermalizes at long times is the basis of thermodynamics but is very difficult to justify mathematically. The main problem is to define what is meant by a ``generic state'' and a ``generic system''. 

A system, whether classical or quantum, is defined by its set of observables and the dynamical law. Most commonly, one considers Hamiltonian dynamics with a time-independent Hamiltonian so that energy is conserved  automatically. Another option is to study Floquet dynamics which only has discrete time-translation symmetry and no conserved quantities, in general. In the case of classical Hamiltonian  systems with a compact phase space, there are negative results showing that a system which is generic in a topological sense is not ergodic for energies near the minimal energy \cite{MarkusMeyer}.\footnote{A Hamiltonian dynamics which is not ergodic has observables other than energy which are integrals of motion and thus is not thermalizing unless one modifies the microcanonical ensemble by taking into account the additional integrals of motion.} The relevance of this result to foundations of thermodynamics is questionable, since thermodynamics concerns itself with large-volume many-body systems  whose energy scales as the volume. For infinite-volume many-body systems,  whether classical or quantum, there are no known obstructions to generic ergodicity.\footnote{In the quantum case, one is naturally led to study infinite-volume many-body systems, since the dynamics of quantum systems with a finite-dimensional Hilbert space is always quasi-periodic.} 

We observe that in most of the physical literature a ``generic dynamical law'' means something rather different from that in \cite{MarkusMeyer}. An obvious obstruction to thermalization in a many-body system is the existence of local observables whose dependence on time is periodic  or periodic modulo some symmetry transformation. One may call such systems ``irregular''. physical intuition suggests that systems which are regular (i.e. not irregular) are thermalizing, and one may regard regularity as a physicist's version of genericity. Whether regular many-body systems are generic in a topological sense is an interesting but logically distinct question.

There is another issue which is much less discussed: the meaning of a ``generic initial state''. Even for classical systems with a compact energy surface, the Poincar\'{e} recurrence theorem shows that one cannot expect thermalization for an arbitrary initial state. For such systems, one typically restricts to states whose probability density is an integrable function (technically, the corresponding measure is absolutely continuous with respect to the Liouville measure). By definition, a mixing dynamical law thermalizes any such state. For many-body systems, the need to restrict the set of initial states is even more acute, since in the infinite-volume limit the set of all states is immense and includes states which cannot be prepared by any conceivable experimental procedure. But in the many-body case it would be wrong to restrict oneself to absolutely continuous states. For example, consider an infinite  system of classical spins evolving under a Hamiltonian dynamics corresponding to a nontrivial finite-range Hamiltonian. A product state with a continuous probability distribution for each spin  is not absolutely continuous with respect to a Gibbs state, yet one expects the former to approach the latter for long times. Similarly, a Gibbs state of an infinite classical spin system with a nontrivial finite-range Hamiltonian is not absolutely continuous with respect to the infinite-temperature state, yet one expects the former to evolve into the latter under a sufficiently generic Floquet dynamics. 

In this paper we describe an infinite set of genuinely many-body classical Floquet dynamical systems of algebraic origin (essentially, classical analogs of quantum spin chains with Clifford dynamics) for which one can verify that the physical intuition is correct. Namely, we prove that for these systems regularity implies thermalization of a large and fairly natural class of initial states. This class includes all Gibbs states of sufficiently uniform differentiable Hamiltonians.

The organization of the paper is as follows. In Section \ref{sec:setup} we describe the class of systems we are interested in. In Section \ref{sec:results} we formulate our results. Sections \ref{sec:proofs1} and \ref{sec:proofs2} contain the proofs. In  Section \ref{sec:discussion} we compare our results with their quantum counterparts \cite{KapRad}.

Some of the proofs in the paper were obtained with the help of ChatGPT 5.4 Pro and Claude Opus 4.6. This work was supported in part by the Simons Investigator Award and the U.S.\ Department of Energy, Office of Science, Office of High Energy Physics, under Award Number DE-SC0011632.

\section{Algebraic Floquet systems}\label{sec:setup}

The systems we are interested in are classical analogs of spin chains. Their phase space has the form
\begin{align}
\A=\prod_{n\in\ZZ} \A_n,
\end{align}
where each $\A_n$ is a $2s$-dimensional torus with the standard symplectic structure. We interpret the index $n$ as labeling the sites of a spatial lattice $\ZZ\subset\RR$. The spatial translation symmetry acts on $\A$ by shifting the index $n$. We parameterize $\A_n$ by the usual $2\pi$-periodic coordinates $\varphi_n^\alpha$, $\alpha\in\{1,\ldots,2s\}$ such that the symplectic form is
$\sum_{a=1}^s d\varphi_n^{2a-1}\wedge d\varphi_n^{2a}.$ 
For any subset $\Lambda\subset\ZZ$ we denote
\begin{align}
    \A_\Lambda=\prod_{n\in\Lambda}\A_n.
\end{align}

$\A$ is infinite-dimensional yet compact when endowed with the product topology. It has a maximum entropy state (i.e. a probability measure)  $\mu_0$ with a constant probability density. This state is known as the Haar state. It is the unique probability measure  invariant under shifts of $\varphi^\alpha_n$, $n\in\ZZ$, $\alpha\in\{1,\ldots,2s\}$, by arbitrary real constants. Equivalently, $\A$ is an abelian group, and $\mu_0$ is invariant under the action of $\A$ on itself. $\mu_0$ is also invariant under the spatial translation $\CT:\varphi^\alpha_n\mapsto\varphi^\alpha_{n-1}$. 

A local observable is a continuous function $f:\A\ra\CC$ which depends only on a finite number of coordinates $\varphi^\alpha_n$. The space of local observables is  dense in $C(\A)$, the space of continuous functions on $\A$. We will sometimes call elements of  $C(\A)$ quasilocal observables. Note that for any two smooth local observables $f,g$ there is a well-defined local observable $\{f,g\}$ (the Poisson bracket of $f$ and $g$). The Poisson bracket makes the commutative algebra of smooth local observables into a Poisson algebra.

We are going to study the dynamics generated by a continuous map $\CF:\A\ra\A$. We impose the following conditions on $\CF$:
\begin{enumerate}
    \item $\CF$ is invertible;
    \item Translational symmetry: $\CF$ commutes with $\CT$;
    \item Locality: if $f\in C(\A)$ depends only on the coordinates on $\A_n$, then $f\circ \CF^{-1}$ depends only on the coordinates on $\A_{[n-D,n+D]}$, where $D\geq 0$ can be chosen the same for all $n$ and $f$;
    \item $\CF$ maps smooth local observables to smooth local observables and preserves the Poisson bracket on the algebra of smooth local observables;
    \item $\CF$ is a homomorphism of abelian groups.
\end{enumerate}
Conditions (1-4) are physically natural. Condition (5) is imposed to make the problem manageable. 

Conditions (2) and (3) imply that $\CF$ is determined by a local update rule: a smooth function $F:\A_{[-D,D]}\ra\A_0$. Condition (5) then implies that $F$ has the form
\begin{align}
    F(\varphi_{-D},\ldots,\varphi_D)^\alpha=\sum_{n=-D}^D \sum_{\beta=1}^{2s} F^\alpha_{\beta,n}\varphi^\beta_n,
\end{align}
where $F^\alpha_{\beta,n}\in\ZZ$. It is convenient to introduce a matrix function
\begin{align}
F(u)^\alpha_\beta=\sum_{n\in\ZZ} F^\alpha_{\beta,n}u^{-n}.
\end{align}
Its matrix elements are Laurent polynomials with integer coefficients. Then $\CF$ can be concisely written as a multiplication of a vector-valued formal power series $\bphi(u)=\sum_{n\in\ZZ} \varphi_n u^n$ by the  matrix-valued Laurent polynomial $F(u)$:
\begin{align}\label{eq:F}
    \CF:\bphi(u)\mapsto F(u)\bphi(u).
\end{align}
In this notation, translation symmetry acts by $\bphi(u)\mapsto u\bphi(u)$ and the fact that it commutes with $\CF$ is obvious. Invertibility of $\CF$ is equivalent to the invertibility of the matrix $F(u)$ as an element of the algebra of $2s\times 2s$ matrices whose matrix elements are Laurent polynomials with integer coefficients.

Dually, one may consider the action of $\CF$ on local observables, and more specifically on characters of $\A$ which form a basis in the space of local observables. Characters are labeled by an infinite-component integer vector with components $q^\alpha_n$ such that only a finite number of components are nonzero. We encode this vector into a Laurent polynomial $\bq(u)^\alpha=\sum_n q^\alpha_n u^n$ whose coefficients are $\ZZ^{2s}$-valued. The character corresponding to $\bq$ is a local observable
\begin{align}
    \chi_\bq=e^{i(\bq,\bphi)},
\end{align}
where $(\bq,\bphi)=\sum_{n,\alpha} \bq_n^\alpha \varphi_n^\alpha$. The symplectomorphism $\CF$ acts on characters as follows:
\begin{equation}
    \chi_\bq\circ \CF^{-1}=\chi_{L\bq},
\end{equation}
where $L(u)=\left(F(u^{-1})^{-1}\right)^T$. Note that $L(u)$ is also an invertible element of the algebra of Laurent polynomials with coefficients in $\Mat(2s,\ZZ)$. 

It is convenient for what follows to introduce the ring $R=\ZZ[u,u^{-1}]$ and a rank-$2s$ free module $R^{2s}$. Then $\bq\in R^{2s}$ and $F,L\in \GL (2s,R)\subset \Mat(2s,R)$. The algebra $\Mat(2s,R)$ is equipped with an involutive anti-automorphism $M(u)\mapsto M(u)^\dagger=M(u^{-1})^T$, in terms of which the relation between $F$ and $L$ can be written as $L=\left(F^\dagger\right)^{-1}$. 

The Poisson bracket of characters is given by
\begin{align}
    \{\chi_\bq,\chi_{\bq'}\}=-(\bq,\Omega\bq')\chi_{\bq+\bq'},
\end{align}
where 
\begin{align}
    \Omega={\mathbf 1}_s\otimes\begin{pmatrix} 0 & 1\\ -1 & 0\end{pmatrix} .
\end{align}
In terms of the matrix-valued function $L$, Condition (4) is equivalent to
\begin{align}\label{eq:pseudoU}
 L^\dagger\Omega L=\Omega.
\end{align}
We will call such $L$ pseudo-unitary. We will see that only invertibility is needed to prove the results below, so from now on we will drop Condition (4) and assume that $L\in \GL (r,R)$ where $r$ is not necessarily even. Pseudo-unitarity is only important for the physical interpretation of our results. 

In summary, we will study Floquet dynamics on $\A$ induced by $L\in \GL(r,R)$. Such dynamical systems bear close similarity to abelian cellular automata as well as to translationally-invariant Clifford Quantum Cellular Automata \cite{CQCA1}. In fact, when $L$ is pseudo-unitary, the corresponding dynamical system can be viewed as a classical limit of a sequence of Clifford QCAs, as discussed in Section \ref{sec:discussion}.

\section{Results}\label{sec:results}

Now we define a class of $\CF$ for which thermalization can be proved. 

\begin{definition}
    Fix a norm $\|\cdot\|$ on $\ZZ^{r}\otimes \RR=\RR^{r}$. For any $\bq\in\ZZ^{r}[u,u^{-1}]$ we let
    \begin{align}
        \|\bq\|_\infty=\sup_{n\in\ZZ} \|\bq_n\|,
    \end{align}
    where $\bq_n$ is the coefficient of $u^n$ in $\bq$. 
\end{definition}

\begin{definition}
    Let $\CF$ be an automorphism of $\A$ specified by a matrix $L\in \GL(r,R)$. We say that $\CF$ has the frequency blowup (FB) property if 
    \begin{align}    \lim_{k\ra\infty}\|L^k\bq\|_\infty=\infty,\quad\forall \bq\neq 0.
    \end{align}
\end{definition}

Next, we define a rather general class of initial states.
For any $\Lambda\subset\ZZ$, let  $\mu_\Lambda(\cdot\vert \bphi_{\Lambda^c})$ be the conditional probability measure $\mu$ on $\A_\Lambda$ induced by $\mu$. In general, it may not exist. If it exists, it is not unique. However, for a fixed $\mu$ any two conditional probability measures regarded as functions on $\A_{\Lambda^c}$ may differ only on a subset of $\mu$-measure zero. We make the following definition.
\begin{definition}
    Let $\mu$ be a probability measure on $\A$ such that for any $n\in\ZZ$ the conditional measure $\mu_{\{n\}}(\varphi_n|\bphi_{\ZZ\backslash\{n\}})$ exists. Define
    \begin{equation}
        \psi_\mu(q)=\sup_{n\in\ZZ}\, {{\rm ess}\sup}_{\bphi_{\ZZ\backslash
        \{n\}}} \left|\int \chi_q(\varphi_n) d\mu_{\{n\}}(\varphi_n|\bphi_{\ZZ\backslash\{n\}})\right|,\quad q\in \ZZ^{r}.
    \end{equation}
    We say that $\mu$ satisfies the Uniform Riemann-Lebesgue (URL) condition if 
    \begin{equation}
        \lim_{q\ra\infty} \psi_\mu(q)=0. 
    \end{equation}
\end{definition}
The meaning of the URL condition is that the single-site conditional probability measures  are sufficiently smooth in a way that is uniform across configurations and sites. 
If the variables on individual sites are independent and identically distributed, then the Riemann-Lebesgue lemma implies that the URL condition is satisfied if the single-site measure is absolutely continuous with respect to the Haar measure on $\A_0$. In general, the URL condition is stronger than the absolute continuity of the single-site conditional measures. As another example, let $\mu$ be a  Gibbs probability measure associated to an energy  function $H(\bphi)$. By definition, this means that for any finite subset $\Lambda\subset\ZZ$, the conditional measure $\mu_\Lambda$ satisfies the Dobrushin-Lanford-Ruelle condition \cite{simongases}:
\begin{equation}
d\mu_\Lambda(\bphi_\Lambda\vert\bphi_{\Lambda^c})=Z_{\Lambda}(\bphi_{\Lambda^c})^{-1} \exp(-\beta H(\bphi_\Lambda,\bphi_{\Lambda^c})) d\mu_0(\bphi_\Lambda).
\end{equation}
Such a measure $\mu$ satisfies the URL condition if
\begin{equation}
\sup_n\sup_{\bphi}\|\partial_{\varphi_n}H\left(\varphi_n,\bphi_{\ZZ\backslash\{n\}}\right)\|<\infty .
\end{equation}

Our first result is
\begin{theorem}\label{thm:FBimpliestherm}
    Let $\CF$ be a symplectomorphism of the form (\ref{eq:F}). If $\CF$ has the FB property and $\mu$ is a probability measure satisfying the URL condition, then for $f\in C(\A)$ we have
    \begin{equation}\label{eq:FBweakconvergence}
       \lim_{n\ra\infty}\mu(f\circ\CF^{-n})=\mu_0(f).
    \end{equation}
\end{theorem}
There is an easily checked sufficient condition for the FB property.
\begin{definition}
    $\CF$ is locally hyperbolic if there exists  $u_0\in S^1$ such that all the eigenvalues of $L(u_0)$ are not on the unit circle. 
\end{definition}
We prove
\begin{theorem}\label{thm:locHimpliesFB}
    If $\CF$ is locally hyperbolic, then it has the FB property.
\end{theorem}
Thus every locally hyperbolic $\CF$ thermalizes any state satisfying the URL condition, and in particular every Gibbs state of a sufficiently uniform differentiable Hamiltonian.

physical intuition suggests that thermalization should apply if there are no quasilocal observables which behave in a periodic manner. To make this precise, we make the following definition.
\begin{definition}
    $\CF$ is irregular iff there exists a nonconstant quasilocal observable $f$ and integers $k,\ell$, $k\neq 0$, such that $f\circ\CF^{-k}=f\circ\CT^{-\ell}$. $\CF$ is regular iff it is not irregular. 
\end{definition}
For an irregular $\CF$, there are nontrivial quasilocal observables which are periodic up to spatial translations. If the initial state is chosen to be translationally invariant, the expectation value of some observables is periodic in time, hence such states do not thermalize. Contrariwise, physical intuition suggests that a regular $\CF$ will thermalize any sufficiently nice initial state.

It easily follows from the results of Rokhlin \cite{Rokhlin1949,Schmidtdynamical} that if $\CF$ is regular, then the corresponding dynamics is ergodic and mixing with respect to the Haar state. However, thermalization is a stronger property than mixing. We prove:

\begin{theorem}\label{thm:no-resonance}
    $\CF$ has the FB property if and only if it is regular.
\end{theorem}
This result, combined with Theorem \ref{thm:FBimpliestherm}, validates the physical intuition about thermalization in the special  case of algebraic Floquet systems.

\section{The frequency blowup property and thermalization}\label{sec:proofs1}

In this section we prove Theorems  \ref{thm:FBimpliestherm} and \ref{thm:locHimpliesFB}. This requires little more than basic analysis and algebra. The proof of Theorems \ref{thm:no-resonance} is more involved and is deferred to the next section.

\begin{proof}[Proof of Theorem \ref{thm:FBimpliestherm}]

For any $\Lambda\subset\ZZ$ we can write $\A=\A_\Lambda\times\A_{\Lambda^c}$. We denote by $\mu_\Lambda$ the marginal probability density on $\A_\Lambda$ corresponding to $\mu$. The group of characters of $\A$ factorizes as a product of the group of characters of $\A_\Lambda$ and the group of characters of $\A_{\Lambda^c}$. For any $\bq\in\ZZ^{r}[u,u^{-1}]$ we denote by $\bq_\Lambda$ its component in the group of characters of $\A_\Lambda$. Clearly, for any $\Lambda$ and $\bq$ we have $\chi_\bq=\chi_{\bq_\Lambda}\chi_{\bq_{\Lambda^c}}$. 

Let $f=\chi_\bq$ be a nontrivial character of $\A$. We have
\begin{align}
    \mu(f\circ\CF^{-k})=\int_\A \chi_{L^k\bq} d\mu.
\end{align}
For any $k$ there exists $n_k$ such that $\|(L^k\bq)_{n_k}\|=\|L^k\bq\|_\infty$. 
Conditioning on the complement of $\{n_k\}$, we get
\begin{align}
    \mu(f\circ\CF^{-k})=\int_{\A_{ \{n_k\}^c}} \chi_{(L^k\bq)_{\{n_k\}^c}}\left[\int_{\A_{\{n_k\}}}\chi_{(L^k\bq)_{\{n_k\}}} d\mu(\varphi_{n_k}|\bphi_{\{n_k\}^c})\right] d\mu_{\{n_k\}^c}.
\end{align}
Therefore
\begin{align}
    \left|\mu(f\circ\CF^{-k})\right|\leq \left|\int_{\A_{\{n_k\}}}\chi_{(L^k\bq)_{\{n_k\}}} d\mu(\varphi_{n_k}|\bphi_{\{n_k\}^c})\right|\leq\psi_\mu((L^k\bq)_{\{n_k\}}).
\end{align}
Now let us consider the limit $k\ra\infty$. By our choice of $n_k$ and the FB condition, we have
\begin{align}
    \limsup_{k\ra\infty} \left|\mu(f\circ\CF^{-k})\right|\leq\limsup_{q\ra\infty}\psi_\mu(q).
\end{align}
 By the URL condition, the expression on the right-hand side is zero. Thus
\begin{align}
    \lim_{k\ra\infty}\mu(f\circ\CF^{-k})=0.
\end{align}
Since by the Stone-Weierstrass theorem, characters are dense in $C(\A)$, and since $|\mu(f\circ\CF^{-k})|\leq \|f\|_\infty$ for any $f\in C(\A)$,  (\ref{eq:FBweakconvergence}) holds for any $f\in C(\A)$.
    
\end{proof}

Now we prove Theorem \ref{thm:locHimpliesFB}.

\begin{proof}[Proof of Theorem \ref{thm:locHimpliesFB}]

Since roots of unity are dense on the unit circle and $L(u)$ is a continuous function of $u$ for $u\neq 0$, we may assume that $u_0$ is a root of unity $\zeta\in\CC$. We denote by $\ZZ[\zeta]\subset\CC$ the ring of polynomials in $\zeta$ with integer coefficients.

Let $\overline{\QQ}\subset \CC$ be the algebraic closure of $\QQ$. It is a subfield of $\CC$ consisting of all roots of all monic polynomials with coefficients in $\QQ$. Note that $\zeta\in\overline{\QQ}$. By assumption, the matrix 
\begin{align}
A:=L(\zeta)\in \GL\left(r,\ZZ[\zeta]\right)\subset \GL(r,\CC)
\end{align}
has no eigenvalue on the unit circle, i.e.\ every eigenvalue $\lambda$ of $A$ satisfies $|\lambda|\neq 1$.

Fix $0\neq \bq\in R^r$. Put $t:=u-\zeta$ and 
expand $\bq(\zeta+t)$ as a power series in $t$ with coefficients in $\ZZ[\zeta]\subset\overline{\QQ}$. Thus, we regard $\bq(\zeta+t)$ as living in $\ZZ[\zeta][[t]]^r$, the rank-$r$ free module over the ring of formal power series in $t$ with coefficients in $\ZZ[\zeta]$.
Let $\ell\ge 0$ be smallest integer such that the coefficient of $t^\ell$ is nonzero, thus
\begin{equation}\label{eq:q-jet-zeta}
\bq(\zeta+t)=t^\ell v + t^{\ell+1}(\cdots),\qquad 0\neq v\in \ZZ[\zeta]^r.
\end{equation}
Similarly, $L(\zeta+t)\in \Mat(r,\ZZ[\zeta][[t]])$ and
\begin{align}
L(\zeta+t)=A+tB(t)
\quad\text{for some }B(t)\in \Mat(r,\ZZ[\zeta][[t]]).
\end{align}
Hence $L(\zeta+t)\equiv A \pmod t$, so $L(\zeta+t)^k\equiv A^k \pmod t$.
Multiplying \eqref{eq:q-jet-zeta} by $L(\zeta+t)^k$ and taking the $t^\ell$--coefficient gives
\begin{equation}\label{eq:jet-evolves-zeta}
[t^\ell]\left(L(\zeta+t)^k\,\bq(\zeta+t)\right)=A^k v.
\end{equation}

Let $\chi_A(T):=\det(T I-A)\in \ZZ[\zeta][T]\subset \overline{\mathbb Q}[T]$ be the characteristic polynomial of $A$.
Since $\chi_A$ is monic and its coefficients are algebraic integers, every eigenvalue $\lambda$ of $A$ is an algebraic integer
(i.e. a root of a monic polynomial with integer coefficients). Since $\det A\neq 0$, all eigenvalues are nonzero.

Choose an eigenvalue $\lambda$ such that the component of $v$ in the generalized eigenspace for $\lambda$ is nonzero. If $e\ge 1$ is the algebraic multiplicity of $\lambda$ in $\chi_A$, set
\begin{align}
g(T):=\frac{\chi_A(T)}{(T-\lambda)^e}\in \overline{\mathbb Q}[T].
\end{align}
In Jordan form, $g(A)$ annihilates all generalized eigenspaces $V_\mu(A)$ with $\mu\neq\lambda$ and is invertible on $V_\lambda(A)$,
hence
\begin{equation}\label{eq:w-def}
w:=g(A)v \neq 0
\quad\text{and}\quad
(A-\lambda I)^m w=0\text{ for some }m\ge 1,
\end{equation}
so $w$ lies in the generalized eigenspace of $A$ for $\lambda$.
Note $w\in \overline{\mathbb Q}^r$ because $A,v$ have entries in $\overline{\mathbb Q}$ and $g$ has coefficients in $\overline{\mathbb Q}$.

\smallskip
If $|\lambda|>1$, set $\lambda':=\lambda$.
If $|\lambda|<1$, let $m_\lambda(x)=x^d+a_{d-1}x^{d-1}+\cdots+a_0\in\ZZ[x]$ be the minimal polynomial of $\lambda$ over $\mathbb Q$.
Since $\lambda$ is an algebraic integer and $\lambda\neq 0$, we have $a_0\in\ZZ\setminus\{0\}$ and
\begin{align}
|a_0|=\prod_{\rho}|\rho(\lambda)|,
\end{align}
where $\rho$ runs over the $d$ complex roots (conjugates) of $m_\lambda$.
Because one factor is $|\lambda|<1$ and the product is an integer $\ge 1$, there exists a conjugate $\lambda'$ of $\lambda$
with $|\lambda'|>1$.

Now use the standard fact (see e.g. \cite{LangAlgebra}): for any conjugate $\lambda'$ of $\lambda$ there exists an automorphism
$\sigma\in \mathrm{Aut}(\overline{\mathbb Q}/\mathbb Q)$ such that
\begin{equation}\label{eq:sigma-lambda}
\sigma(\lambda)=\lambda'.
\end{equation}
Define
\begin{align}
\zeta':=\sigma(\zeta).
\qquad
A':=\sigma(A).
\end{align}
Since $\sigma$ is an automorphism, $\zeta'$ is still a root of unity. 
Because $A=L(\zeta)$ and $L$ has integer Laurent coefficients (fixed by $\sigma$), applying $\sigma$ entrywise gives
\begin{equation}\label{eq:A0-is-eval}
A'=\sigma(L(\zeta))=L(\sigma(\zeta))=L(\zeta').
\end{equation}
Also set $v':=\sigma(v)$ and $w':=\sigma(w)\neq 0$.
Applying $\sigma$ to \eqref{eq:w-def} yields
\begin{align}
(A'-\lambda' I)^m w' = 0,
\end{align}
so $w'$ lies in the generalized eigenspace of $A'$ for eigenvalue $\lambda'$.

Since $|\lambda'|>1$ and $w'\neq 0$ lies in that generalized eigenspace, there exist constants $c>0$ and an integer $s\ge 0$ (can be nonzero if the relevant Jordan block has size larger than $1$) such that
\begin{equation}\label{eq:exp-growth-w0}
\|(A')^k w'\|\ \ge\ c\,k^s|\lambda'|^k
\qquad\text{for all $k$ sufficiently large.}
\end{equation}
Moreover, $w'=\sigma(g(A))\,v'$ and $\sigma(g(A))$ is a fixed matrix, so for a constant $C_g:=\|\sigma(g(A))\|$ we have
\begin{equation}\label{eq:exp-growth-vprime}
\|(A')^k v'\|\ \ge\ \frac{1}{C_g}\,\|(A')^k w'\|
\ \ge\ \frac{c}{C_g}\,k^s|\lambda'|^k
\qquad\text{for all $k$ sufficiently large.}
\end{equation}

Apply $\sigma$ to \eqref{eq:q-jet-zeta} (it fixes $t$ and acts on coefficients) to get
\begin{align}
\bq(\zeta'+t)=t^\ell v' + t^{\ell+1}(\cdots).
\end{align}
As before, $L(\zeta'+t)\equiv A'\pmod t$, hence taking the $t^\ell$--coefficient gives
\begin{equation}\label{eq:jet-evolves-u0}
[t^\ell]\left(L(\zeta'+t)^k\,\bq(\zeta'+t)\right)=(A')^k v'.
\end{equation}

Define $\bq^{(k)}(u)=L(u)^k \bq(u)$. Then
\begin{equation}\label{eq:jet-lower}
\left\|[t^\ell](\bq^{(k)}(\zeta'+t))\right\|=\|(A')^k v'\|
\ \ge\ c_1\, k^s \Lambda^k,
\quad \forall k\gg 1,
\end{equation}
with $\Lambda:=|\lambda'|>1$ and $c_1:=c/C_g>0$.

Write $\bq^{(k)}(u)=\sum_j \bq^{(k)}_j u^j$. Taylor-expanding $(\zeta'+t)^j$ gives
\begin{align}
(\zeta'+t)^j=\sum_{m\ge 0}\binom{j}{m}(\zeta')^{\,j-m}t^m,
\end{align}
where $\binom{j}{m}$ is the usual (generalized) binomial coefficient, valid for all $j\in\ZZ$.
Taking the $t^\ell$-coefficient gives 
\begin{align}
[t^\ell]\left(\bq^{(k)}(\zeta'+t)\right)=\sum_{j} \bq^{(k)}_j\,\binom{j}{\ell}\,(\zeta')^{\,j-\ell}.
\end{align}
Since $|\zeta'|=1$, we have
\begin{equation}\label{eq:jet-upper-raw}
\left\|[t^\ell](\bq^{(k)}(\zeta'+t))\right\|
\ \le\ \sum_j \|\bq^{(k)}_j\|\,\left|\binom{j}{\ell}\right|
\ \le\ \|\bq^{(k)}\|_\infty\sum_{j\in {\rm supp}\, \bq^{(k)}} \left|\binom{j}{\ell}\right|,
\end{equation}
where ${\rm supp}\,\bq^{(k)}\subset\ZZ$ is the set of degrees of monomials of the Laurent polynomial $\bq^{(k)}(u)$.

Let $D\ge 0$ be such that the degrees of the matrix Laurent polynomial $L(u)$ are contained in $[-D,D]$.
If $\bq(u)$ is supported in $[a,b]\subset\ZZ$, then $\bq^{(k)}(u)=L(u)^k \bq(u)$ is supported in
$[a-kD,b+kD]$, so the sum in \eqref{eq:jet-upper-raw} runs over at most $O(k)$ indices $j$ with $|j|\leq C' k$, where $C'$ is a positive constant which depends on $\bq$ and $D$. Moreover, for fixed $\ell$ there is a constant $C_\ell>0$ such that for all $j\in\ZZ$,
$|\binom{j}{\ell}|\le C_\ell (1+|j|)^\ell$.
Hence, there is a constant $C'_\ell>0$ (depending on $\ell$, $D$, and $\bq$) such that
\begin{equation}\label{eq:binom-sum}
\sum_{j\in{\rm supp}\, \bq^{(k)}} \left|\binom{j}{\ell}\right|\ \le\ C'_\ell\, (1+k)^{\ell+1}\qquad \forall k\ge 1.
\end{equation}
Combining \eqref{eq:jet-upper-raw} and \eqref{eq:binom-sum} gives
\begin{equation}\label{eq:exp-upper}
\|[t^\ell](\bq^{(k)}(\zeta'+t))\|\ \le\ C'_\ell\,(1+k)^{\ell+1}\,\|\bq^{(k)}\|_\infty,\qquad \forall k\geq 1.
\end{equation}

Combining (\ref{eq:jet-lower}) and (\ref{eq:exp-upper}) we get
\begin{align}
\|\bq^{(k)}\|_\infty\ \ge\ \frac{c_1}{C'_\ell}\,\frac{k^s\Lambda^k}{(1+k)^{\ell+1}} \qquad (k\gg 1).
\end{align}
Since $\Lambda>1$, this proves
$\|\bq^{(k)}\|_\infty\to\infty$ as $k\to\infty$.
\end{proof}

\section{Regularity and the frequency blowup property}\label{sec:proofs2}

In this section we prove Theorem \ref{thm:no-resonance}. We start with the following algebraic characterization of regularity.
\begin{lemma}\label{lem:algebraicregularity}
    $\CF$ is regular if and only if for all $m\in\ZZ$ and all $k\in\NN$ the equation  $L^k\bq=u^m\bq$ for $\bq\in R^{r}$ has no nonzero solutions. 
\end{lemma}
\begin{proof}
    Suppose there is a nonzero $\bq\in R^{r}$ such that $L^k\bq=u^m\bq$ for some $k,m$, $k\neq 0$. Then $f=\chi_\bq$ is a nonconstant local observable which satisfies $f\circ\CF^{-k}=f\circ\CT^m$, hence $\CF$ is irregular. 

    Conversely, suppose $\CF$ is irregular and $f\in C(\A)$ is nonconstant and solves $f\circ\CF^k=f\circ\CT^m$. Let $c_\bq$ be the Fourier coefficient of $f$ corresponding to $\bq\in R^{r}$. Since $C(\A)\subset L^2(\A)$, Parseval's  theorem gives
    \begin{align}\label{eq:parseval}
        \sum_{\bq\in R^{r}}\left|c_{\bq}\right|^2=\int_\A |f|^2<\infty.
    \end{align}
    On the other hand, the equation for $f$ implies 
    \begin{align}
        c_{u^{-m}L^k\bq}=c_{\bq},\quad\forall \bq\in R^{r}.
    \end{align}
    Pick a nonzero $\bq_0\in R^{r}$ such that $c_{\bq_0}\neq 0$, and consider the orbit of $\bq_0$ under the action of the Laurent polynomial $P=u^{-m}L^k$. The Fourier coefficients are constant along the orbit. If the orbit is infinite, then there is an infinite sequence in $R^{r}$  along which the Fourier coefficients $c_{\bq}$ are equal to $c_{\bq_0}\neq 0$. This contradicts (\ref{eq:parseval}). Hence the orbit of $\bq_0$ must be finite and there exists $\ell>0$ such that $P^\ell\bq_0=\bq_0$, i.e. $L^{k\ell}\bq_0=u^{m\ell}\bq_0$. 
\end{proof}

Recall that we equipped $M=R^r$ with a norm $\|\cdot\|_\infty$. We will repeatedly use the following property of this norm.
\begin{lemma}\label{lem:norm}
    For any $T=\sum_{m\in\ZZ} T_m u^m\in \End_R(M)$, we have 
    \begin{align}
        \|T\bq\|_\infty\leq \|\bq\|_\infty \sum_m\|T_m\|.
    \end{align}
\end{lemma}
\begin{proof}
    Straightforward.
\end{proof}

In what follows it will be useful to regard
$L\in \End_R(M)$ as an endomorphism of the vector space $V=M\otimes K\simeq K^r$ over the field $K=\QQ(u)$ (the field of fractions of $R=\ZZ[u^{\pm 1}])$. The field $K$ is not algebraically closed, so the Jordan normal form theorem does not apply to operators on $V$. Its place is taken by the primary decomposition theorem (see e.g. \cite{DummitFoote}).
\begin{lemma}\label{lem:primary-proj}
Let
\begin{align}
  \chi(u,t)=\det(t I-L(u))\in R[t]
\end{align}
and write, in $K[t]$,
\begin{align}
  \chi(u,t)=\prod_{i=1}^s \chi_i(u,t)^{e_i}
\end{align}
where the $\chi_i(u,t)\in R[t]$ are distinct monic polynomials which are irreducible in $K[t]$.
Define
\begin{align}
  V_i := \ker\!\left(\chi_i(u,L)^{e_i}:V\to V\right)\subseteq V.
\end{align}
Then:

\smallskip
\noindent\textup{(1)} $V=\bigoplus_{i=1}^s V_i$.

\smallskip
\noindent\textup{(2)} There exist $K$-linear $L$-equivariant projections $\pi_i\in\End_K(V)$ such that
\begin{align}
  \pi_i^2=\pi_i,\qquad \pi_i\pi_j=0\ (i\neq j),\qquad \sum_{i=1}^s \pi_i=\mathrm{Id}_V,\qquad \pi_i(V)=V_i,
\end{align}
and each $\pi_i$ commutes with $L$.

\smallskip
\noindent\textup{(3)} For each $i$ there exists $0\neq d_i(u)\in R$ such that
\begin{align}
  T_i := d_i(u)\,\pi_i
\end{align}
satisfies $T_i(M)\subseteq M$. In particular, for $\bq\in M$ the element
\begin{align}
  w_i:=T_i(\bq)\in M\cap V_i,
\end{align}
and $w_i=0$ if and only if $\pi_i(\bq)=0$ in $V$.

\end{lemma}

\begin{proof}
View $V$ as a $K[t]$-module by letting $t$ act as $L$.
Cayley--Hamilton gives $\chi(u,L)=0$ on $V$, hence $V$ is a finitely generated torsion $K[t]$-module.
Since $K[t]$ is a PID, the primary decomposition theorem for modules over a PID yields the decomposition
$V=\bigoplus_i \ker(\chi_i(u,L)^{e_i})$ and the corresponding $L$-equivariant idempotent projections arising from the Chinese remainder theorem in $K[t]/(\chi)$.

For (3), identify $V\cong K^r$ and $M\cong R^r$ with $M$ embedded in $V$ by $m\mapsto m\otimes 1$.
Each $\pi_i$ is represented by a matrix in $\Mat(r,K)$, so choosing a nonzero common denominator $d_i(u)\in R$ for its entries yields
$T_i=d_i\pi_i\in\Mat(r,R)$ and hence $T_i(M)\subseteq M$.
\end{proof}

Next we recall a few facts about Mahler measures, see e.g. \cite{Schmidtdynamical}. The multiplicative Mahler measure ${\mathfrak M}(p)$ of a one-variable complex polynomial $p(z)=a_0 (z-\alpha_1)\ldots(z-\alpha_d)$ is defined as ${\mathfrak M}(p)=|a_0|\prod_i {\rm max}(1,|\alpha_i|)$. By Jensen's formula \cite{Ahlfors}, ${\mathfrak M}(p)=\exp(\mm(p))$, where
\begin{align}\label{eq:aMahler}
\mm(p)=\frac{1}{2\pi}\int_0^{2\pi}\log|p\left(e^{i\theta}\right)|d\theta.
\end{align}
In what follows it will be more convenient to work with $\mm(p)$. We will refer to it as the additive Mahler measure since  $\mm(pp')=\mm(p)+\mm(p')$ for any two polynomials $p,p'$. 

If $p$ has integral coefficients, then obviously ${\mathfrak M}(p)\geq 1$ and thus $\mm(p)\geq 0$. A theorem of Kronecker implies that $\mm(p)=0$ if and only if $p$ is a product of cyclotomic polynomials times a power of $z$. This generalizes straightforwardly to Laurent polynomials with integer coefficients.

Even more generally, one can define the additive Mahler measure of an $\ell$-variable Laurent polynomial $p(z_1,\ldots,z_\ell)$ to be 
\begin{align}
\mm(p)=\int_{\substack{|z_1|=1\\ \ldots\\ |z_\ell|=1}}\log|p(z_1,\ldots,z_\ell)|\frac{dz_1\ldots dz_\ell}{(2\pi i)^\ell z_1\ldots z_\ell}.
\end{align}
If $p$ has integral coefficients, then one can show that  $\mm(p)\geq 0$. Furthermore, Boyd \cite{Boyd} characterized integer-coefficient Laurent polynomials $p$ such that $\mm(p)=0$. Namely, $\mm(p)=0$ if and only if 
\begin{align}
p(z_1,\ldots,z_\ell)=z_1^{A_1}\ldots z_\ell^{A_\ell}\prod_{i=1}^N \Phi_{n_i}(z_1^{p_{i,1}}\ldots z_\ell^{p_{i,\ell}}),
\end{align}
where $A_1,\ldots,A_\ell,p_{i,1},\ldots,p_{i,\ell}$, $i\in \{1,\ldots,N\},$ are integers and $\Phi_n(z)$, $n\geq 2$, is the $n^{\rm th}$ cyclotomic polynomial, i.e. the unique monic polynomial of a variable $z$ whose roots are exactly the primitive $n^{\rm th}$ roots of unity. The following lemma can be regarded as a converse to Boyd's theorem in the two-variable case.

\begin{lemma}\label{lem:mm-expanding-root}
Let $p(u,t)\in K[t]$, $K=\QQ(u)$, be an irreducible monic polynomial in $t$ of degree $d$. Suppose the two-variable additive Mahler measure of $p$ is positive,  $\mm(p)>0$.
Then there exist $\varepsilon>0$, $\theta_0\in[0,2\pi)$, and an open interval $I\ni\theta_0$ such that
for each $\theta\in I$ the one-variable polynomial $p(e^{ i\theta},t)$ has a \emph{simple} root
$\lambda(\theta)$ depending real-analytically on $\theta$ and satisfying
\begin{align}
|\lambda(\theta)|\ge 1+\varepsilon,\quad \forall\theta\in I.
\end{align}
\end{lemma}

\begin{proof}
For fixed $\theta$, let $\lambda_1(\theta),\dots,\lambda_d(\theta)$ be the roots of $p(e^{ i\theta},t)$.
By Jensen's formula applied in the $t$-variable and the fact that $p$ is monic in $t$,
\begin{align}
   \mm(p)=\frac{1}{2\pi}\int_0^{2\pi} \sum_{j=1}^d \log^+|\lambda_j(\theta)|\,d\theta, 
\end{align}
where $\log^+x=\max(\log x,0).$
Since $\mm(p)>0$, there exist $\varepsilon>0$ and a set $E$ of positive measure such that
$\max_j|\lambda_j(\theta)|>1+2\varepsilon$ for all $\theta\in E$. Since $p(u,t)$ is irreducible, its discriminant with respect to $t$ is a nonzero element of $K=\QQ(u)$, hence can vanish only for a finitely many values of $\theta$. We choose 
$\theta_0\in E$ such that the discriminant of $p(e^{i\theta},t)$ is nonzero.
Then all roots of $p(e^{ i\theta_0},t)$ are simple, and at least one root $\lambda(\theta_0)$ satisfies
$|\lambda(\theta_0)|>1+2\varepsilon$.
By the implicit function theorem, this root extends to a real-analytic root $\lambda(\theta)$
on a neighborhood $I$ of $\theta_0$; shrinking $I$ if necessary gives $|\lambda(\theta)|\ge 1+\varepsilon$ on $I$.
\end{proof}

The next lemma is standard.
\begin{lemma}\label{lem:jordan-growth}
Let $A\in \Mat(r,\CC)$ and let $\lambda$ be an eigenvalue of $A$ with $|\lambda|>1$.
If $v\in\CC^r$ belongs to the generalized eigenspace of $\lambda$, then there exist
$m\ge 1$ and $c>0$ such that for all $k\ge 1$,
\begin{align}
\|A^k v\|\ \ge\ c\,k^{m-1}|\lambda|^k.
\end{align}
\end{lemma}
\begin{proof}
Put $A$ in Jordan normal form. On a Jordan block $J=\lambda I+N$ of size $\ell$,
\begin{align}
J^k=\lambda^k\sum_{j=0}^{\ell-1}\binom{k}{j}\lambda^{-j}N^j.
\end{align}
If $m$ is the largest integer such that $v$ has a component with $N^{m-1}v\neq 0$, then $\|J^k v\|=O( k^{m-1}|\lambda|^k)$.
Summing over blocks gives the stated estimate.
\end{proof}

\begin{proof}[Proof of Theorem \ref{thm:no-resonance}]
First suppose $\CF$ has the FB property. If $\CF$ were irregular, then by Lemma~\ref{lem:algebraicregularity}
there would exist $k\in\NN$, $m\in\ZZ$, and $0\neq \bq_0\in M$ such that
\begin{equation}\label{eq:resonant-q0}
L^k\bq_0=u^m\bq_0.
\end{equation}
Write $n=qk+r$ with $0\le r<k$. Using $R$-linearity of $L$ we get
\begin{align}
L^n\bq_0=L^{qk+r}\bq_0=L^r(L^{qk}\bq_0)=L^r(u^{mq}\bq_0)=u^{mq}L^r\bq_0.
\end{align}
Multiplication by $u^{mq}$ shifts coefficients and hence preserves $\|\cdot\|_\infty$, so
\begin{align}
\|L^n\bq_0\|_\infty=\|L^r\bq_0\|_\infty\le \max_{0\le r<k}\|L^r\bq_0\|_\infty<\infty.
\end{align}
In particular $\|L^n\bq_0\|_\infty$ does not tend to $\infty$, contradicting the FB property.
Thus $\CF$ must be regular.

\smallskip

Conversely, we prove the contrapositive: assume $\CF$ does \emph{not} have the FB property.
Then there exists $0\neq \bq\in M$ such that $\|L^k\bq\|_\infty\nrightarrow\infty$.
Hence there exist $B<\infty$ and a strictly increasing sequence $k_j\to\infty$ with
\begin{equation}\label{eq:subseq-bdd-q}
\|L^{k_j}\bq\|_\infty\le B\qquad\forall j.
\end{equation}

Let $K=\QQ(u)$, $V:=M\otimes_R K\cong K^r$, and let
\begin{align}
\chi(u,t)=\det(tI-L(u))\in R[t].
\end{align}
Factor $\chi$ in $K[t]$ as in Lemma~\ref{lem:primary-proj}:
\begin{align}
\chi(u,t)=\prod_{i=1}^s \chi_i(u,t)^{e_i},
\qquad
V_i:=\ker(\chi_i(u,L)^{e_i}:V\to V),
\qquad
V=\bigoplus_{i=1}^s V_i.
\end{align}
Let $\pi_i$ and $T_i=d_i(u)\pi_i$ be the $L$-equivariant maps from Lemma~\ref{lem:primary-proj}(2)--(3).
Define
\begin{align}
w_i:=T_i(\bq)\in M\cap V_i.
\end{align}
Since $\sum_i\pi_i=\mathrm{Id}_V$ and $\bq\neq 0$ in $V$, we have $w_i\neq 0$ for at least one index $i$.

Moreover, $T_i$ commutes with $L$, and by Lemma~\ref{lem:norm} there exists a constant $C_i>0$ (depending only on $T_i$)
such that for all $j$,
\begin{equation}\label{eq:subseq-bdd-wi}
\|L^{k_j}w_i\|_\infty
=\|T_i(L^{k_j}\bq)\|_\infty
\le C_i\,\|L^{k_j}\bq\|_\infty
\le C_i B.
\end{equation}

\smallskip
\noindent\textit{Claim:} If $\mm(\chi_i)>0$, then $w_i=0$.

\smallskip
\noindent
Assume $\mm(\chi_i)>0$ and $w_i\neq 0$. Set $A(\theta):=L(e^{i\theta})$ and $\widehat{w_i}(\theta):=w_i(e^{i\theta})$.
Let $D\ge 0$ be such that every entry of $L(u)$ is supported in degrees contained in $[-D,D]$.
If the support of $w_i$ is contained in some integer interval of length $S_0$, then $L(u)^k w_i$ is supported in an interval
of length at most $S_0+2Dk$. Therefore for every $\theta$ and every $j$,
\begin{align}\label{eq:polynomial-bound}
\|A(\theta)^{k_j}\widehat{w_i}(\theta)\|
=\|(L^{k_j}w_i)(e^{i\theta})\|
\le (S_0+2Dk_j+1)\,\|L^{k_j}w_i\|_\infty
\le C'_i(1+k_j),
\end{align}
for a constant $C'_i>0$ (using \eqref{eq:subseq-bdd-wi}).
Thus we have a polynomial upper bound along the subsequence $k_j$.

On the other hand, since $\chi_i$ is irreducible and monic in $t$ and $\mm(\chi_i)>0$,
Lemma~\ref{lem:mm-expanding-root} gives an open interval $I$ and a real-analytic simple root
$\lambda(\theta)$ of $\chi_i(e^{i\theta},t)$ on $I$ with $|\lambda(\theta)|\ge 1+\varepsilon$.

Pick $\theta_0\in I$ and put $u_0:=e^{i\theta_0}$.
Shrinking $I$ if needed, choose a simply connected neighborhood $U\subset\CC^\times$ of $u_0$ such that
$\Delta_i(u)\neq 0$ on $U$, where $\Delta_i$ is the $t$-discriminant of $\chi_i(u,t)$.
Then $\chi_i(u,t)$ has $d=\deg_t \chi_i$ holomorphic root functions
$\Lambda_1(u),\dots,\Lambda_d(u)$ on $U$ with
\begin{align}
\chi_i(u,t)=\prod_{j=1}^d (t-\Lambda_j(u))\quad\text{in }\mathcal O(U)[t],
\end{align}
and (after relabeling) $\Lambda_1(e^{ i\theta})=\lambda(\theta)$ for $\theta\in I$.

Let $F/K$ be a splitting field of $\chi_i$.
The field $K(\Lambda_1,\dots,\Lambda_d)\subset \mathcal M(U)$ (meromorphic functions on $U$) is also a splitting field,
hence $K$-isomorphic to $F$; fix such an isomorphism
\begin{align}
\iota:F \xrightarrow{\ \cong\ } K(\Lambda_1,\dots,\Lambda_d)\subset \mathcal M(U).
\end{align}
For $a\in F$ we view $\iota(a)$ as a meromorphic function on $U$, and for $\theta\in I$ we define
\begin{align}
a(\theta):=\iota(a)\left(e^{i\theta}\right),
\end{align}
shrinking $I$ if necessary to avoid poles for the finitely many coefficients used below.

Let $m:=e_i$.
In $F[t]$ we have $\chi_i(u,t)^m=\prod_{j=1}^d (t-\lambda_j)^m$ for the abstract roots $\lambda_j\in F$,
ordered so that under $\iota$ the root $\lambda_1$ corresponds to the branch $\Lambda_1$ on $U$.
By the B\'{e}zout identity for the PID $F[t]$, there exists $\epsilon(t)\in F[t]/(\chi_i^m)$ such that
\begin{align}
\epsilon(t)\equiv 1\ \ \mathrm{mod}\ (t-\lambda_1)^m,
\qquad
\epsilon(t)\equiv 0\ \ \mathrm{mod}\ \prod_{j=2}^d (t-\lambda_j)^m.
\end{align}
Choose a polynomial representative $\epsilon(t)=\sum_{j=0}^J \epsilon_j t^j$ with $\epsilon_j\in F$.
For $\theta\in I$ define $\epsilon_\theta(t):=\sum_j \epsilon_j(\theta)t^j\in\CC[t]$ and set
\begin{align}
P(\theta):=\epsilon_\theta\left(A(\theta)\right).
\end{align}
By construction, $P(\theta)$ commutes with $A(\theta)$ and acts as the identity on the generalized eigenspace of $A(\theta)$
for the eigenvalue $\lambda(\theta)$ and as $0$ on the other generalized eigenspaces inside
$\ker(\chi_i(e^{i\theta},A(\theta))^m)$. Let $v(\theta):=P(\theta)\hat w_i(\theta)$. This is real-analytic function on $I$ with values in $\CC^r$ whose value at each $\theta\in I$ lies in the generalized eigenspace of $A(\theta)$ corresponding to the eigenvalue $\lambda(\theta)$.

If $v(\theta_*)\neq 0$ for some $\theta_*\in I$, then Lemma~\ref{lem:jordan-growth} yields constants $c>0$ and $\ell\ge 0$
such that for all $k\ge 1$,
\begin{align}
\|A(\theta_*)^k v(\theta_*)\|\ \ge\ c\,k^\ell |\lambda(\theta_*)|^k.
\end{align}
In particular this holds for $k=k_j$, and since $|\lambda(\theta_*)|>1$ this contradicts the polynomial upper bound along $k_j$,
because $A(\theta_*)^{k_j}v(\theta_*)=P(\theta_*)A(\theta_*)^{k_j}\widehat{w_i}(\theta_*)$ and $P(\theta_*)$ has finite operator norm. Hence $v(\theta)=0$ for all $\theta\in I$.

Now consider $z:=\epsilon(L)w_i\in V\otimes_K F\cong F^r$.
Under $\iota$, $z$ becomes a vector of meromorphic functions on $U$, and for $\theta\in I$ its value at
$u=e^{i\theta}$ is exactly $v(\theta)=0$.
Therefore all coordinates of $\iota(z)$ vanish on an arc, hence $\iota(z)=0$ in $\mathcal M(U)^r$ and thus $z=0$ in $F^r$.
So
\begin{align}
\epsilon(L)w_i=0\quad\text{in }V\otimes_K F.
\end{align}

Let $G=\Gal(F/K)$ and fix $\sigma_1,\dots,\sigma_d\in G$ with $\sigma_j(\lambda_1)=\lambda_j$.
Set $\epsilon_j:=\sigma_j(\epsilon)\in F[t]/(\chi_i^m)$.
Then $\epsilon_j$ is the idempotent corresponding to the factor $(t-\lambda_j)^m$, and
$\sum_{j=1}^d \epsilon_j \equiv 1 \pmod{\chi_i^m}$.
Since $L$ and $w$ are defined over $K$, applying $\sigma_j$ to the identity $\epsilon(L)w_i=0$ yields
$\epsilon_j(L)w_i=0$ for all $j$.
Hence
\begin{align}
w_i=\left(\sum_{j=1}^d \epsilon_j(L)\right)w_i=\sum_{j=1}^d \epsilon_j(L)w_i=0,
\end{align}
a contradiction. Thus $\mm(\chi_i)>0$ implies $w_i=0$.

Since there exist $i$ with $w_i\neq 0$, we must have $\mm(\chi_i)=0$ for at least one such index $i$. Fix such an index and denote $p(u,t):=\chi_i(u,t)$.
Then $p(u,t)$ is irreducible in $K[t]$, monic in $t$, and $\mm(p)=0$.

By Boyd's characterization of zero Mahler measure (as recalled above), the irreducible two-variable Laurent polynomial $p$
has the form
\begin{align}
p(u,t)=c\,u^e t^f\,\Phi_d(u^a t^b)
\end{align}
for integers $a,b,e,f$, a cyclotomic polynomial $\Phi_d$, and $c\neq 0$.
Since $\chi(u,0)=\det(-L(u))\neq 0$, no irreducible factor of $\chi$ can be divisible by $t$,
hence $f=0$. Since $p$ has positive $t$-degree, we also have $b\neq 0$.

Let $\lambda\in\overline{\QQ(u)}$ be a root of $p(u,t)$ (here $\overline{\QQ(u)}$ denotes the algebraic closure of $\QQ(u)$). Then $\Phi_d(u^a\lambda^b)=0$, so $(u^a\lambda^b)^d=1$ and hence
$\lambda^{bd}=u^{-ad}$. Setting $k:=|b|d>0$ and $m:=-{\rm sign}(b)ad\in\ZZ$, we obtain $\lambda^k=u^m$.
Since $p\mid\chi$, the number $\lambda$ is an eigenvalue of $L(u)$, so $u^m$ is an eigenvalue of $L(u)^k$, and therefore
\begin{align}
\det\left(L(u)^k-u^m I\right)=0.
\end{align}
Thus there exists $0\neq v\in K^r$ with $(L^k-u^m I)v=0$.
Clearing denominators yields $0\neq \bq_0\in M$ satisfying $L^k\bq_0=u^m\bq_0$.
By Lemma~\ref{lem:algebraicregularity}, this implies that $\CF$ is irregular.

\end{proof}

\section{Quantum versus classical Floquet spin chains}\label{sec:discussion}

The classical Floquet systems studied in this paper bear similarity to translationally-invariant (TI) Clifford QCAs acting on a quantum spin chain. Such a quantum spin chain has an on-site Hilbert space of dimension $p^s$. It is a tensor product of $s$ copies of an irreducible $p$-dimensional unitary representation of the generalized Pauli group with generators $X,Z$ satisfying 
\begin{align}
    X^p=Z^p=1,\quad XZ=e^{2i\pi/p}ZX.
\end{align}
Monomials in the generalized Pauli matrices $X^a_n,Z^a_n$, $a\in\{1,\ldots,s\}$, $n\in\ZZ$ form a basis for the algebra of local observables of the spin chain. A TI Clifford QCA is a translationally-invariant QCA which maps this distinguished basis to itself, up to scalar factors. One can show that such QCAs can be described by pseudo-unitary elements of $GL(2s,R_p)$, where $R_p$ is the ring of Laurent polynomials with coefficients in $\ZZ_p$. Given an algebraic Floquet system described by $L\in GL(2s,R)$ one can produce a TI Clifford QCA by letting $L_p=L\bmod p$. 

This map from the set of algebraic Floquet systems to the set of TI Clifford QCAs has a simple physical meaning.  Recall that quantizing a $2s$-dimensional torus equipped with a symplectic structure $\frac{p}{2\pi} \sum_{a=1}^s d\phi^{2a-1}\wedge d\phi^{2a}$ gives a quantum torus whose algebra of observables has dimension $p^{2s}$. Accordingly, quantization of $\A$ gives a quantum spin chain whose on-site algebra of observables is the complex matrix algebra of dimension $p^{2s}$. While quantization is not a functor, in general, in this particular case there is a natural way to quantize the automorphism $\CF$ of the form (\ref{eq:F}). One can show that the corresponding $*$-automorphism of the quantum spin chain is the TI QCA specified by $L_p(u)=L(u) {\bmod p}$. 

It is natural to ask how the ergodic properties of the symplectomorphism $\CF$ are related to those of the corresponding TI Clifford QCA. The task is somewhat simplified by the fact that for systems of this type ergodicity and mixing are identical. In the classical case this is a well-known result \cite{Rokhlin1949,Schmidtdynamical}, for the quantum case see \cite{KapRad}. One can show that the TI Clifford QCA corresponding to a pseudo-unitary $L_p\in GL(2s,R_p)$ is ergodic/mixing iff the equation $L_p^k\bq=\bq$ does not admit nontrivial solutions for all $k\in\NN$. 

It is easy to see that classical ergodicity does not imply quantum ergodicity. It is equally easy to see that quantum ergodicity even for a single value of $p$ implies ergodicity for the ``parent'' classical problem. 

The situation with thermalization is more complicated. Regularity of a TI QCA $\alpha$ can be defined as the absence of nontrivial local observables $a$ satisfying $\alpha^k(a)=\tau^m(a)$ where $\tau$ is the translation automorphism. One can show that a TI Clifford QCA is regular iff the equation $(L^k-u^m)\bq=0$, $\bq\in R_p^{2s}$, does not have nonzero solutions. Hence if a regular TI Clifford QCA $\alpha$ is a quantization of an algebraic Floquet symplectomorphism $\CF$, then $\CF$ is also regular. However, the relation between regularity and thermalization on the quantum level is not as straightforward as on the classical level. First of all, the conditions on quantum initial states which ensure thermalization are  rather restrictive, see \cite{KapRad} for details. In particular, it is not known if regular TI Clifford QCAs thermalize quantum Gibbs states. Second, the quantum thermalization statement itself is weaker: at present, it is only proved that expectation values approach their infinite-temperature values after one excludes a zero-density subset of times \cite{KapRad}. 

The reason for this is that the mechanism for thermalization in the quantum case is very different from the classical one. Since the single-site Hilbert space is finite-dimensional, frequency blowup is impossible in the quantum case. Instead, quantum thermalization relies on ``diffusion'': the unbounded growth of support of all local observables \cite{PV1,ETDS,KapRad}. This is a much more delicate property than frequency blowup. Despite this difference, both in the quantum and classical cases the physical intuition correctly predicts the circumstances in which the dynamical law is thermalizing.

\bibliographystyle{unsrt}
\bibliography{Bibliography.bib}

\end{document}